\documentclass[12pt]{amsart}
\usepackage{geometry} 
\usepackage{graphicx}
\usepackage{amsfonts}
\usepackage{amsthm}
\usepackage{amsmath}
\usepackage{amssymb}
\usepackage[arrow,matrix]{xy}
\newtheorem{theorem}{Theorem}
\newtheorem{corollary}{Corollary}

\newtheorem{proposition}{Proposition}
\theoremstyle{definition}
\newtheorem{example}{Example}

\newcommand{\co}{\colon\,}
\newcommand{\bT}{\mathbb T}
\newcommand{\bR}{\mathbb R}
\newcommand{\bC}{\mathbb C}
\newcommand{\bZ}{\mathbb Z}
\newcommand{\bP}{\mathbb P}
\newcommand{\ubZ}{\underline{\mathbb Z}}
\newcommand{\bN}{\mathbb N}
\newcommand{\bQ}{\mathbb Q}
\newcommand{\SO}{\mathop{\rm SO}}

\newcommand{\SL}{\mathop{\rm SL}}

\newcommand{\pt}{\text{pt}}
\newcommand{\lp}{\textup{(}}
\newcommand{\rp}{\textup{)}}
\newcommand{\Ext}{\operatorname{Ext}}
\newcommand{\Hom}{\operatorname{Hom}}
\newcommand{\Ann}{\operatorname{Ann}}
\newcommand{\Aut}{\operatorname{Aut}}
\newcommand{\coker}{\operatorname{coker}}
\newcommand{\im}{\operatorname{im}}

\title{K-theoretic matching of brane charges
in S- and U-duality}
\author{Stefan Mendez-Diez}
\address{Department of Mathematical and Statistical Sciences\\
University of Alberta\\
Edmonton, AB T6G 2G1, Canada}
\email[Stefan Mendez-Diez]{mendezdi@ualberta.ca}
\thanks{The Research of the first author was partially supported by NSF grants DMS-0504212 and
  DMS-0805003. Some of this work constituted part of the
  author's Ph.D.\ dissertation, submitted to the University of
  Maryland in May, 2010.}
  
\author{Jonathan Rosenberg}
\address{Department of Mathematics\\
University of Maryland\\
College Park, MD 20742-4015, USA} 
\email[Jonathan Rosenberg]{jmr@math.umd.edu}
\thanks{The research of the second author was partially supported by NSF grants DMS-0504212 and DMS-0805003.}
  
  
\begin{document}

\begin{abstract}
We discuss $K$-theoretic matching of D-brane charges in the
string duality between type I on $\bT^4$ and type IIA on $K3$. This
case is more complex 
than the familiar case of IIA/IIB duality, which is already well
understood, but it turns out that replacing $K3$ by
its orbifold blow-down seems largely to resolve the apparent problems
with the theory.
\end{abstract}

\maketitle

\section{Introduction}

It is believed that the five superstring theories are all related
through various dualities: T-duality, S-duality, and a combination of
both, known as U-duality \cite{MR1321523}.
Sometimes the explicit dualities between
the different superstring theories are unclear. 

We can determine a lot about possible dualities by looking at stable
D-brane configurations. If two theories are dual to one another then
they should have the same spectra. Therefore,  dual string theories 
should have equivalent brane configurations. (\emph{Equivalent}
here means, for instance, that it should be possible to match up the
brane charges in the two theories, and thus these charges should
live in isomorphic groups.) The stable D-brane
configurations in a given theory depend only on the topology of the
spacetime and can be classified by [possibly twisted]
$K$-theory \cite{Moore:2004,
  Witten:1998, Witten:2000, Olsen:1999}.  The $K$-theoretic
classification of $D$-brane charges has proven very useful in the
study of string dualities, particularly $T$-duality between the type
IIA and IIB theories. By putting together the known duality between
type I and the $\SO(32)$ heterotic string 
theories with the one between the $\SO(32)$ heterotic and type-IIA string
theories, one obtains an example of a U-duality between type-I and type-IIA
string theories. This paper will focus on matching the $K$-theoretic
classification of $D$-brane charges in this example of $U$-duality.

It is conjectured that type-IIA string theory on $K3$ is dual to the
$\SO(32)$ heterotic string on the $4$-torus, $\bT^4$
\cite{Hull:1995}, \cite[\S4]{MR1334520}, \cite{MR1349795},
\cite{MR1394147}, \cite{MR1731643}, \cite[p.\ 424]{Becker:2007}, \cite{Aharony:2003}, \cite{Kiritsis:2000}. The
$\SO(32)$ heterotic string is 
believed to be equivalent to type-I string theory 
via S-duality \cite{MR1381609}, \cite{Tseytlin:1996},
\cite{MR1352442}, so this gives 
a duality between type-I string theory on $\bT^4$ and type-IIA string
theory on $K3$.  (This chain of equivalences is mentioned explicitly
in \cite[p.\ 258]{MR2010972}.) Not only does this give a concrete
example of a duality between the type I and type IIA theories, but it
can be extended to relate all five of the superstring theories. At the
$\bT^4/\bZ_2$ point of the moduli space, the type IIA theory can be
related by a $T$-duality to the type IIB theory on $\bT^4/\bZ_2$ with
an $NS5$-brane at each of the $16$ singularities \cite{Aharony:2003,
  Kutasov:1995, MR1408165}. This is also related to the type I theory
on $\bT^4$ via $S$-duality \cite{MR1408165}. Furthermore the $\SO(32)$
and $E_8\times E_8$ heterotic theories compactified on a torus are
$T$-dual to one another \cite[p.\ 288]{Becker:2007}.

Stable D-brane charges  are classified by $KO(X)$ in type-I string
theory \cite{Witten:1998, Reis:2006}, and by $\widetilde{K}(X)$ (resp.,
$K^{-1}(X)$) in type-IIB theory (resp., type-IIA theory)
\cite{Witten:1998, Witten:2000, Becker:2007}. Since it is 
conjectured that type-I string theory on $\bT^4$ and type-IIA string
theory on $K3$ are dual to each other, they should have the
same possible charges. The puzzle is that $KO^*(\bT^4)$ contains
$2$-torsion (implying the existence of torsion charged $D$-branes in
the type I theory compactified on $\bT^4$) and $K^0(K3)\cong
\bZ^{24}$, while $K^{-1}(K3)\cong 0$, so that there wouldn't appear to
be any stable D-brane charges at all in type-IIA 
theory compactified on $K3$! Even in IIB theory on $K3$, it appears there is no
room for torsion brane charges!

Now one could object that $D$-branes in one theory do not
necessarily transform to $D$-branes in a dual theory. However, if
there is a $D$-brane in one theory, there must be an object 
(not necessarily a D-brane) in the dual
theory that corresponds to the $D$-brane with the correct charge. For
example, under the duality between the type IIA theory on $K3$ and
the $\SO(32)$ heterotic theory on $\bT^4$, the fundamental heterotic
string appears as a non-singular soliton solution in the type IIA
theory. Similarly, the fundamental type IIA string appears as a
non-singular soliton solution in the heterotic theory
\cite{MR1349795}. The spectrum of the elementary heterotic string
contains charged BPS states, but the spectrum of the elementary type
IIA string does not. All of the gauge fields in the type IIA theory
come from the Ramond-Ramond (RR) sector, so the only charged BPS
states are $D$-branes. The elementary type IIA string is not
charged. This is explained by the fact that the charged states in the
type IIA theory appear as solitons in the heterotic theory and not in
the elementary string spectrum. The fundamental type IIA string
appears as an uncharged soliton in the heterotic theory. Even though
$D$-branes do not necessarily transform to $D$-branes, it is possible
to match the charges in the dual theories. 
In our current example there is no way to explain 
the torsion charges that appear in the type I theory in the type IIA
theory other than as torsion charged $D$-branes. The $NS5$-brane that
wraps the $K3$ appearing in the type IIA theory corresponds to the
fundamental string in the heterotic theory. Since the fundamental
heterotic string is BPS, it cannot have torsion charge. Therefore, the
$NS5$-brane occurring in the type IIA theory cannot have torsion
charge and the torsion charged branes occurring in the type I theory
must map to torsion charged $D$-branes in the type IIA theory. 
There must be $D$-branes in the type IIA theory that carry
torsion charge, but there is no torsion in the cohomology or [untwisted]
$K$-theory of $K3$. Furthermore, twisting by an $H$-flux cannot introduce
torsion in the twisted $K$-theory of $K3$ since any twist would have
to live in $H^3(K3)=0$. We show that if one first removes 
the sixteen isolated singularity points of an orbifold blow-down of $K3$,
such an isomorphism is close to being achieved, albeit in a very
nontrivial way. Thus 
this calculation provides an interesting test of S- and U-duality. We
then look at the K-theory classification of $D$-brane charges at the
orbifold point, $\bT^4/\bZ_2$, of the moduli space of $K3$ for both
the type IIA and IIB theories and discuss some of the issues that
arise. 

The first author would like to thank Chuck Doran for many useful
conversations and suggestions.


\section{$KO^*(\bT^4)$}

We work throughout with $K$-theory with compact support. Thus for a
locally compact space $X$ which is not compact, $KO^*(X)$ is
(essentially by definition) identified with $\widetilde{KO}^*(X^+)$,
where $X^+=X \cup \{\infty\}$ is the one-point compactification of $X$
(e.g., $(\bR^n)^+ = S^n$). $KO^{-n}(\bT^4)$ can be computed from 
$$KO^{-i}(\pt)\cong  \left \{ \begin{array}[pos]{l}\bZ,\;\;\;i\equiv 0 \pmod{4} \\
\bZ_2,\;\;i\equiv 1,2 \pmod{8} \\
0,\;\;\;\text{otherwise},\\
\end{array} \right. $$
by iterating the formula $KO^k(X\times S^1)\cong KO^k(X) \oplus
KO^{k-1}(X)$, which follows from the axioms of a (generalized)
cohomology theory. Thus we obtain:
\[
\begin{aligned}
KO^{-i}(\bT^4)&\cong KO^{-i}(\bT^3)\oplus KO^{-(i+1)}(\bT^3)\\
&\cong KO^{-i}(\bT^2)\oplus 2KO^{-(i+1)}(\bT^2)\oplus
KO^{-(i+2)}(\bT^2)\\
&\cong KO^{-i}(\bT)\oplus 3KO^{-(i+1)}(\bT)\oplus 3KO^{-(i+2)}(\bT) \oplus
KO^{-(i+3)}(\bT)\\ 
&\cong KO^{-i}(\pt)\oplus 4KO^{-(i+1)}(\pt)\oplus
6KO^{-(i+2)}(\pt)\oplus 4KO^{-(i+3)}(\pt)\\
&\qquad \qquad \oplus KO^{-(i+4)}(\pt). 
\end{aligned}
\]
Since type-I string theory is a ten-dimensional theory, the actual
spacetime manifold for type-I string theory compactified on $\bT^4$ 
is $\bT^4\times\bR^6$. Stable D-brane charges in type-I string theory
on $\bT^4\times\bR^6$ are thus classified by
$$KO^0(\bT^4\times\bR^6)\cong KO^{-6}(\bT^4)\cong 6\bZ\oplus
5\bZ_2.$$ 

However, this may not be the end of the story. If 
$\iota\co Y^{p+1}\hookrightarrow \bT^4\times\bR^6$ is the 
inclusion of a (proper)
D$p$-brane in $\bT^4\times\bR^6$, with $Y$ assumed to be spin for anomaly
cancellation, the Gysin map in $KO$-theory gives a map
$\iota_!\co KO(Y) \to KO^{3-p}(\bT^4)$ obtained as the following
composite:
\begin{multline*}
KO(Y) \xrightarrow{\text{P.D.}} KO_{p+1}(Y) \xrightarrow{\iota_*}
KO_{p+1}(\bT^4\times\bR^6) \\ \xrightarrow{(\text{P.D.})^{-1}}
KO^{10-(p+1)}(\bT^4\times\bR^6) \cong KO^{3-p}(\bT^4).
\end{multline*}
Here $\hbox{P.D.}$ denotes the Poincar\'e duality isomorphism and
$(\text{P.D.})^{-1}$ is its inverse.
A Chan-Paton bundle with orthogonal gauge group gives a class in
$KO(Y)$, and thus via the Gysin map $\iota_!$ a D-brane charge in
$KO^{3-p}(\bT^4)$. This is $KO^{-6}(\bT^4)\cong 6\bZ\oplus
5\bZ_2$ when $p=9$ or $1$. (Recall that real $K$-theory satisfies Bott
periodicity with period $8$.) Similarly, a Chan-Paton bundle with
symplectic gauge group gives a class in $KSp(Y)\cong
KO^4(Y)$ (since real and symplectic $K$-theory agree after a
dimension shift by $4$), and thus via the Gysin map $\iota_!$ a
D-brane charge in $KO^{7-p}(\bT^4)$. This can again be identified with
$KO^{-6}(\bT^4)$ when $p=5$. The $9$-branes and $1$-branes with real
Chan-Paton bundles, along with the $5$-branes with symplectic Chan-Paton
bundles, account for all the usual BPS-branes of type-I superstring
theory \cite[p.\ 223]{Becker:2007}.  But as pointed out by many
authors, e.g., 
\cite{{MR1650234},{MR1660435},{MR1658287},{Witten:1998},{MR1806590},
{Bergman:1999ta},{MR1731758},{MR1915386}}, there can
be additional D-brane charges coming from non-supersymmetric, but
still stable, branes with other values of $p$. Such charges (for 
type-I superstring theory compactified on $\bT^4$) are summarized in the
following Table \ref{table:KOgroups}. The various kinds of branes are
hypothetical; not all of them actually occur. Also note that after
inverting $2$, $KO$ and $KSp$ are the same, so the nature of the
Chan-Paton gauge group only affects the $2$-torsion.

\begin{table}[bh]
\begin{tabular}{||c|c|c|l||}
\hline
$p$&bundle type&BPS?&Charge group\\
\hline
$9$&O&yes&$KO^{-6}(\bT^4)\cong 6\bZ\oplus 5\bZ_2$\\
$9$&Sp&no&$KO^{-2}(\bT^4)\cong 6\bZ\oplus \bZ_2$\\
$8$&O&no&$KO^{-5}(\bT^4)\cong 4\bZ\oplus \bZ_2$\\
$8$&Sp&no&$KO^{-1}(\bT^4)\cong 4\bZ\oplus 5\bZ_2$\\
$7$&O&no&$KO^{-4}(\bT^4)\cong 2\bZ$\\
$7$&Sp&no&$KO^{0}(\bT^4)\cong 2\bZ\oplus 10\bZ_2$\\
$6$&O&no&$KO^{-3}(\bT^4)\cong 4\bZ$\\
$6$&Sp&no&$KO^{1}(\bT^4)\cong 4\bZ\oplus 10\bZ_2$\\
$5$&O&no&$KO^{-2}(\bT^4)\cong 6\bZ\oplus \bZ_2$\\
$5$&Sp&yes&$KO^{2}(\bT^4)\cong 6\bZ\oplus 5\bZ_2$\\
$4$&O&no&$KO^{-1}(\bT^4)\cong 4\bZ\oplus 5\bZ_2$\\
$4$&Sp&no&$KO^{3}(\bT^4)\cong 4\bZ\oplus \bZ_2$\\
$3$&O&no&$KO^{0}(\bT^4)\cong 2\bZ\oplus 10\bZ_2$\\
$3$&Sp&no&$KO^{4}(\bT^4)\cong 2\bZ$\\
$2$&O&no&$KO^{1}(\bT^4)\cong 4\bZ\oplus 10\bZ_2$\\
$2$&Sp&no&$KO^{5}(\bT^4)\cong 4\bZ$\\
$1$&O&yes&$KO^{2}(\bT^4)\cong 6\bZ\oplus 5\bZ_2$\\
$1$&Sp&no&$KO^{6}(\bT^4)\cong 6\bZ\oplus \bZ_2$\\
$0$&O&no&$KO^{3}(\bT^4)\cong 4\bZ\oplus \bZ_2$\\
$0$&Sp&no&$KO^{7}(\bT^4)\cong 4\bZ\oplus 5\bZ_2$\\
$-1$&O&no&$KO^{4}(\bT^4)\cong 2\bZ$\\
$-1$&Sp&no&$KO^{8}(\bT^4)\cong 2\bZ\oplus 10\bZ_2$\\
\hline
\end{tabular}
\medskip
\caption{Groups of D$p$-brane charges for type I compactified on $\bT^4$}
\label{table:KOgroups}
\end{table}


\section{K-Theory of a Desingularized $K3$}
\label{sec:KK3}

It becomes immediately obvious that we don't want to just use $K3$ as a
possible dual topology to type-I theory on $\bT^4$, because the
complex $K$-theory of $K3$ contains no torsion, and $K^{-1}(K3)\cong 0$,
so could not possibly
which would imply there were no stable $D$-branes in the type IIA
theory compactified on $K3$ and we could not possibly explain the
torsion charges occurring in the type I theory. Instead of $K3$, a
$\bZ_2$ orbifold quotient of $\bT^4$, where $\bZ_2$ acts by multiplication by $-1$ (which is a singular limit of
$K3$), is often used in string theory because the Ricci-flat metric can
be explicitly determined \cite[\S9.3]{Becker:2007}. The orbifold,
$\bT^4/\bZ_2$, has $16$ isolated singular points, the $16$
fixed points of the $\bZ_2$ action on $\bT^4$. We can remove
these $16$ singular points by first removing $16$ open balls in
$\bT^4$ surrounding each of the singular points. We then divide out by
the $\bZ_2$ action on $\bT^4$ minus the $16$ open balls to
obtain a smooth manifold (with boundary), $N$. The boundary of $N$ is
$16$ copies of 
$\bR\bP^3$.  $16$ copies of an Eguchi-Hanson space (which
topologically is the unit disk bundle of the tangent bundle of $S^2=\bC\bP^1$)
are usually 
glued onto $N$ along their common boundary to create a manifold with
the same topology as $K3$. This cannot be the correct manifold for our
purposes because
again its cohomology (and thus its $K$-theory) contains no torsion.
 Since the original $K3$ in the type IIA theory has no singularities,
 we do not want to allow for any physical effect from the
 singularities (which like $D$-branes are sources for R-R
 charges). Therefore, physically, we are only interested in fields
 that approach a 
constant value at the singularities, and it makes more sense simply to
collapse the singularities and deal with the singular quotient space
$(\bT^4/\bZ_2) / (\text{singularities}) \cong N/\partial N $.
Since this space is the one-point compactification of the
\emph{interior} of $N$, we are interested in $\widetilde K^*(N/\partial N )
\cong K^*(N, \partial N )$, the relative $K$-theory of the
manifold $N$ rel its boundary. While we don't want to allow for any
effect from the singularities in the type IIA theory, the situation is
different when looking at the type IIB theory.  

The type IIA theory compactified on $\bT^4/\bZ_2$, the orientifold
point of the moduli space of $K3$, is $T$-dual to the type IIB theory
compactified on $\bT^4/\bZ_2$ with an $NS5$-brane and an orientifold
$5$-plane at each of the $16$ fixed points of the $\bZ_2$
action. Performing an $\SL(2,\bZ)$ transformation on this
configuration transforms the $NS5$-branes to $D5$-branes, giving us
the type IIB theory on $\bT^4/\bZ_2$ with a $D5$-brane and an
$O5$-plane at each of the singularities
\cite{Kutasov:1995,Aharony:2003,MR1408165}. The final type IIB
configuration is equivalent to the type I theory on $\bT^4$. This
gives us an explicit way of relating type I on $\bT^4$, type IIA on
$K3$ and type IIB on $\bT^4/\bZ_2$ with either an $NS5$-brane or a
$D5$-brane at each of the $O5$-planes. The $T$-duality transformation
between the type IIA theory and the type IIB theory transforms the
$64$ blow-up modes of of the type IIA theory at the orbifold point of
the moduli space ($4$ for each singularity) into moduli controlling
the positions of the $5$-branes. While we described this duality for
the case when the $5$-branes are located at the singularities, the
duality is true when the branes are moved \cite{Aharony:2003}. With
this in mind, we will move the $5$-branes a small distance from the
singularities and consider them located on the boundary of $N$. In
this case the effects of the boundary are important and we cannot mod
it out. Therefore in the type IIB theory we are interested in
$K^*(N)$. 

To compute $K^*(N,\partial N)$ and $K^*(N)$ we will first need to compute the
homology of $N$.  Let $M$ be $\bT^4\smallsetminus
(16\; \text{open}\; \text{balls})$, which is the double
cover of $N$. Since $N$ is obtained from $M$ by dividing out by a free
$\bZ_2$-action, there is a spectral sequence
$H_p(\bZ_2,H_q(M))\Rightarrow H_{p+q}(N)$. (See for example \cite[Theorem
8$^\text{bis}$.9]{McCleary}.) So we must first 
determine the homology of $M$ as a $\bZ_2$-module. 
The homology of $M$ is torsion-free as a $\bZ$-module,
but we need to compute it as a $\bZ G$-module, where $G= \bZ_2$. Since
$G$ is self-dual, the
ring $\bZ G$ is isomorphic to the representation ring $R$ of $G$
studied in the Appendix (Section \ref{sec:appendix}), so we refer to
results there for more details about homological algebra over this ring.
In particular, it turns out 
(as a consequence of iterated application of Propositions 
\ref{prop:RstrZ}, \ref{prop:RstrZZ}, \ref{prop:extRJbyR}, and
\ref{prop:extRIbyRJ}) that the homology of $M$ is a direct sum of
copies of
three standard $R$-modules (all torsion-free as $\bZ$-modules):
$R$ as a module over itself,
the trivial module $\bZ$ (or $R/I$, in the notation of the Appendix), and
$\bZ$ with the non-trivial $\bZ_2$-action (where the generator of the
group acts by multiplication by $-1$), which we call $\ubZ$ to
distinguish it from $\bZ$ with the trivial $\bZ_2$-action.  (This last
$R$-module, in the notation of the Appendix, is $R/J$.)

First of all, note that the cohomology ring of 
$\bT^4$ is an exterior algebra on $4$ generators. Each of these
generators is sent to its negative under the $\bZ_2$ action, so $\bZ_2$ acts
trivially on the even exterior powers and non-trivially on the odd
exterior powers.  So $H^1(\bT^4)\cong H_1(\bT^4)\cong H^3(\bT^4)\cong
H_3(\bT^4) \cong \ubZ^4$, while $H^2(\bT^4)\cong H_2(\bT^4)\cong \bZ^6$.
Now by a simple transversality argument,
removing $16$ balls from $\bT^4$ does
not change the fundamental group, so $\pi_1(M)\cong\pi_1(\bT^4)\cong\ubZ^4$.
Therefore $H_1(M)\cong \ubZ^4$. To obtain $H_2(M)$ we
can use the Meyer-Vietoris sequence:  
\[
\xymatrix{
H_2(M\cap 16B^4) \ar[r] &
  H_2(M)\oplus H_2(16B^4) \ar[r] & H_2(\bT^4) \ar[r]&
  H_1(M\cap 16B^4)\\ 
0 \ar[r] \ar@{=}[u]& H_2(M)\ar[r] \ar@{=}[u]& \bZ^6 \ar[r] \ar@{=}[u]& 0
\ar@{=}[u]}
\]
Here $B^4$ is the closed $4$-ball, so $M\cap16B^4=16S^3$. So we see
$H_2(M)\cong \bZ^6$. We can determine $H_3(M)$ from the long exact
sequence of pairs, using the pair $(M, \partial M)$, where $\partial M
=16 S^3$. The part of the long exact sequence we are interested in is  
$$H_4(M)\rightarrow H_4(M,\partial M)\rightarrow H_3(16
S^3)\rightarrow H_3(M)\rightarrow H_3(M, \partial M)\rightarrow H_2(16
S^3).$$ 
$H_4(M)\cong 0$ since $M$ has a nonempty boundary. By Poincar\'e duality
$H_4(M, \partial M)\cong H^0(M)\cong H_0(M)\cong \bZ$. And similarly, $H_3(M,
\partial M)\cong H^1(M)\cong FH_1(M)\oplus TH_0(M)\cong \ubZ^4$. Finally,
$H_3(16S^3)\cong \bZ^{16}$ (the $\bZ_2$ action is trivial since it
preserves orientation on $S^3$) and $H_2(16S^3)\cong 0$. Putting this all
together, the long exact sequence becomes 
$$0\rightarrow\bZ\rightarrow\bZ^{16}\rightarrow
H_3(M)\rightarrow\ubZ^4\rightarrow 0,$$
where it's easy to see that the map $\bZ\rightarrow\bZ^{16}$ is just
the diagonal inclusion. 
This shows us that $H_3(M)$ is an extension of $\ubZ^4$ by
$\bZ^{15}$. To summarize, we have
\[
H_i(M)\cong \left\{\begin{array}[pos]{ll} \bZ, & i=0\\
\ubZ^4, & i=1\\
\bZ^6, & i=2\\
\text{extension of }\ubZ^4\text{ by }
\bZ^{15}, & i=3\\
0, & \text{otherwise}.\\
\end{array} \right.
\]

However, it will turn out that the extension in $H_3(M)$ cannot be
split, and rather is $R^4\oplus \bZ^{11}$. To see this, observe that
the spectral sequence $E^2_{p,q}=H_p(\bZ_2$, $H_q(M)) \Rightarrow
H_{p+q}(N)$ has only $4$ rows, and thus $E^5=E^\infty$. 
Recall that
$H_p(\bZ_2,\bZ)\cong \bZ$ for $p=0$, $\bZ_2$ for $p\ge 1$ odd, and
$0$ for $p\ge 2$ even, while $H_p(\bZ_2,\ubZ)\cong \bZ_2$ for $p\ge 0$
even, $0$ for $p$ odd, so the first 3 rows of $E^2$ are:

\begin{table}[h]
\begin{center}
\begin{tabular}{l|lllllll}
$2$ & $\bZ^6$ & $\bZ_2^6$ & 0 & $\bZ_2^6$ & 0 & $\bZ_2^6$ & $\cdots$ \\
$1$ & $\bZ_2^4$ & 0 & $\bZ_2^4$ & 0 & $\bZ_2^4$ & 0 & $\cdots$ \\
$0$ & $\bZ$ & $\bZ_2$ & 0 & $\bZ_2$ & 0 & $\bZ_2$ &$\cdots$ \\ \hline
$q\,/\,p$ & $0$ & $1$ & $2$ & $3$ & $4$ & $5$ & $\cdots$
\label{eq:specseq}
\end{tabular}
\end{center}
\caption{$E^2$ of the spectral sequence for computing $H_*(N)$}
\label{table:specseq}
\end{table}

Since $N$ is a
noncompact $4$-manifold, its homology must vanish in dimension $4$ and
higher, so $E^2_{p,3}$ for $p\ge 1$ must be killed off by entries in
the rows with $q\le 2$ and total degree $p+4\ge 5$. If $H_3(M)$ were
to contain a summand isomorphic to 
$\ubZ$, then $E^2_{p,3}$ would contain $2$-torsion for all even $p$,
which contradicts the fact that $E^2_{p,q}=0$ for $q\le 2$, $p>0$, and
$p+q$ even.  Thus, by iterated application of Propositions
\ref{prop:RstrZ}, \ref{prop:RstrZZ}, \ref{prop:extRJbyR}, and
\ref{prop:extRIbyRJ}, $H_3(M)\cong \bZ^{11}\oplus (\bZ G)^4$. It
follows that the $q=3$ row of $E^2$ in the spectral sequence is
$\bZ^{15}$ for $p=0$, $\bZ^{11}_2$ for $p\ge 1$ odd, and $0$ otherwise.

We can now determine $H_*(N)$ from the spectral sequence 
with $E^2_{p,q}=H_p(\bZ_2$, $H_q(M))$, Table \ref{table:specseq}. 
For $p>0$, $E_{p,q}^2$ is all torsion, so the free part of $H_q(N)$ is
the same as for $E_{0,q}^2$. So we see that the Betti numbers of $N$ are 
$$\beta_i(N)=\left\{\begin{array}[pos]{ll} 1, & i=0\\
0, & i=1\\
6, & i=2\\
15, & i=3\\
0, & \text{otherwise}.\\
\end{array} \right.$$
First we know that $H_0(N)\cong \bZ$ since $N$ is connected. We also know
$H_4(N)\cong 0$ because $N$ has a nonempty boundary. Now $H_3(N)\cong H^1(N,
\partial N)\cong FH_1(N, \partial N)\oplus TH_0(N, \partial N)$, and
$TH_0(N, \partial N)\cong 0$. So $H_3(N)$ is free, and thus isomorphic to
$\bZ^{15}$. We can use Mayer-Vietoris with $N$ and 16 Eguchi-Hanson
spaces, $E$,  
since $N\cup_{16\bR\bP^3}16E \cong K3$, to show that $H_2(N)\cong
\bZ^6$, from the exact sequence
\[
\dots\to H_{k+1}(K3)\to H_k(16\bR\bP^3) \to
  H_k(N)\oplus H_k(16E) \to H_k(K3) \to \dots.
\]
Furthermore, $E$ has the same homotopy type as $S^2$, since it
is the unit disk bundle of the tangent bundle of $S^2$. The part of
the Mayer-Vietoris sequence we are interested in is: 
\[
\xymatrix{
H_2(16\bR\bP^3) \ar[r] \ar@{=}[d] &
  H_2(N)\oplus H_2(16E) \ar[r] \ar@{=}[d] & H_2(K3) \ar@{=}[d] \\ 
0 \ar[r] & H_2(N)\oplus\bZ^{16} \ar[r] & \bZ^{22}.\\
}
\]
From this we see that $H_2(N)$ injects into a free abelian group and thus must
be free. We have shown that only $H_1(N)$ can have any torsion.  

We can calculate $H_1(N)$ from the spectral sequence
$H_*(\bZ_2,H_*(M))\Rightarrow H_*(N)$.  
$$E_{1,0}^2=H_1(\bZ_2, H_0(M))\cong \bZ_2.$$
No non-zero differential hits it or leaves it
because we have a first quadrant spectral sequence.
$$E_{0,1}^2=H_0(\bZ_2, H_1(M))\cong {\bZ_2}^4.$$
Again no differential hits it since $E_{2,0}^2=H_2(\bZ_2,
\bZ)\cong 0$. Therefore $H_1(N)$ is an extension of  $\bZ_2$ by
${\bZ_2}^4$. Also $H_1(N)$ is a quotient of ${\bZ_2 }^{16}$ as can be
seen from Mayer-Vietoris:
\[
\xymatrix{
H_1(16\bR\bP^3) \ar[r] \ar@{=}[d] &
  H_1(N)\oplus H_1(16E) \ar[r] \ar@{=}[d] & H_1(K3) \ar@{=}[d] \\
{\bZ_2}^{16}  \ar[r] & H_1(N)  \ar[r] & 0,\\
}
\]
so all of its torsion is of order 2. Therefore the extension is
trivial and $H_1(N)\cong {\bZ_2}^5$. Putting this all together, we see that
$$
H_i(N)\cong \left\{\begin{array}[pos]{ll} \bZ, & i=0\\ 
{\bZ_2}^5, & i=1\\
\bZ^6, & i=2\\
\bZ^{15}, & i=3\\
0, & \text{otherwise}.\\
\end{array} \right.$$
By Poincar\'e duality for manifolds with boundary (also known as
Alexander-Lef\-schetz duality), the cohomology of $N$ relative to its
boundary is thus
$$
H^i(N,\partial N)\cong H_{4-i}(N) \cong 
\left\{\begin{array}[pos]{ll} 0, & i=0\\ 
\bZ^{15}, & i=1\\
\bZ^6, & i=2\\
{\bZ_2}^5, & i=3\\
\bZ, & i=4\\ 
0, & \text{otherwise}.\\
\end{array} \right.$$

The $K$-theory is then computed from the Atiyah-Hirzebruch spectral
sequence 
\[
H^p(N, K^q(\pt)) \Rightarrow 
K^{p+q}(N), 
\] 
but all differentials vanish since the first differential is the Steenrod
operation $\operatorname{Sq}^3$, which must vanish, and there is no
room in this case for any higher differentials. Since the spectral sequence
collapses at $E_2$ and $\widetilde H^2(N)\cong \bZ^6 \oplus
{\bZ_2}^5$, while $H^1(N)=0$ and $H^3(N)\cong \bZ^{15}$,
\begin{equation}
\label{eq:KN}
\begin{cases}
K^0(N) \cong\bZ^7\oplus\bZ_2^5,\\
K^{-1}(N) \cong\bZ^{15}.
\end{cases}
\end{equation}
The Universal Coefficient Theorem in $K$-theory (see \cite{MR948692}
for references) gives a short exact sequence
\begin{equation}
\label{eqn:UCTZ}
0\to\Ext^1_{\bZ}(K^{k+1}(N),\bZ)\to K_k(N)\to\Hom(K^k(N),\bZ)\to 0,  
\end{equation}
that splits (non-canonically), so $K_0(N)\cong \bZ^7$ and 
\[
K_1(N)\cong \bZ^{15}\oplus \Ext^1_{\bZ}(\bZ^7\oplus\bZ_2^5,\bZ) 
\cong \bZ^{15}\oplus \bZ_2^5.
\]
Finally, since $N$ is an even-dimensional compact spin$^c$ manifold
(with boundary), we have Poincar\'e duality
$$K_*(N)\cong K^*(N,\partial N).$$
So $K^0(N, \partial N) \cong K_0(N)  \cong \bZ^7$ and 
$K^{-1}(N, \partial N)  \cong  K_1(N) \cong \bZ^{15} \oplus{\bZ_2}^5$. 

The stable $D$-brane charges in the type IIB theory on $N\times\bR^6$ with $16$ $5$-branes on its boundary are classified by $K^0(N\times\bR^6)\cong K^0(N)\cong \bZ^7\oplus\bZ_2^5$. In the type IIA theory, the stable D-brane charges on
$\mathring N\times\bR^6$ are classified by 
$K^{-1}\bigl((N/ \partial N)\times\bR^6\bigr)\cong K^{-1}(N/
\partial N) \cong K^{-1}(N,\partial N) \oplus K^{-1}(\pt) \cong
\bZ^{15} \oplus{\bZ_2}^5$ (since $K^{-1}(\pt) \cong 0$).


\section{$D$-Brane Charges in $\bZ_2$-Equivariant $K$-Theory}
\label{sec:equivK}

Since $N$ was obtained from the orbifold $\bT^4/\bZ_2$, we can obtain
more information about the classification of D-branes on
$(N,\partial N)$ by first looking at how D-branes are classified on
$\bT^4/\bZ_2$. 

As described in \cite{Witten:1998},  \cite{Johnson:1997},
\cite{GarciaCompean:1999} and \cite{MR1777343}, stable D-brane
configurations on an orbifold $X/G$ are classified by the
$G$-equivariant $K$-theory $K_G(X)$ in the type IIB theory, and
$K^{-1}_G(X)$ in the type IIA theory. Therefore, to classify stable
D-brane configurations in the type II theories on the orbifold limit
of $K3$, $\bT^4/\bZ_2$, we must compute $K_{\bZ_2}^*(\bT^4)$. 

For the remainder of this paper let $G=\bZ_2$ and
$R=R(G)=\bZ[t]/(t^2-1)$ be the representation ring of $G$, where $t$ is the
nontrivial character of $G$. Let $I=(t-1)$ and $J=(t+1)$. These are 
prime ideals with $R/I\cong R/J\cong\bZ$, and $R_{(I)}\cong
R_{(J)}\cong \bQ$ (see Appendix).  

$K_G(X)$ is an $R$-module, and the $R$-module structure carries more
information than just the abelian group structure. For the sake of
generality we compute $K^*_G(\bT^n)$ as an $R$-module when $G$ acts on
$\bT^n=\bR^n/\bZ^n$ by $-1$ and we place no restriction on $n$. 

\begin{theorem}
\label{thm:equivKTn}
Let $G$ act on $\bT^n=\bR^n/\bZ^n$ via multiplication by $-1$ on
$\bR^n$. Then $K_G^*(\bT^n)$ is entirely concentrated in even degrees,
and $K_G^0(\bT^n)\cong 2^{n-1}\cdot R \oplus 2^{n-1}\cdot (R/J)$.
\end{theorem}
\begin{proof}
Throughout this proof we will be using the result from \cite{MR0234452}:
If $C$ is a closed $G$-invariant subspace of a locally compact
$G$-space $X$ then the sequence 
\begin{equation}
\label{eq:clgsub}
\begin{aligned}
K_G^0(X-C) &\to K_G^0(X) \to K_G^0(C) \to K_G^1(X-C) \to K_G^1(X) \\ &\to
K_G^1(C) \to K_G^0(X-C) \to \cdots 
\end{aligned}
\end{equation}
is exact. Here the last part was gotten using Bott periodicity,
$K_G^2(X-C)\cong K_G^0(X-C)$. We will also use the result from
\cite{0903.1035}: 
\begin{equation}
\label{eq:kgrn}
\begin{array}{rcl}
K_G^0(\pt)=K_G^0(\bR^k) & = & R,\text{ if $k$ is even}\\
K_G^0(\bR^k) & = & R/J,\text{ if $k$ is odd}\\
K_G^1(\pt)=K_G^1(\bR^k) & = & 0.\\
\end{array}
\end{equation}
A fundamental domain for $\bT^n$ is $F=\{
(x_1,\ldots , x_n) : |x_j|\leq\frac{1}{2}\}/\sim$, where
$-\frac{1}{2}\sim\frac{1}{2}$. Define: 
\begin{multline}
Y_k=\bigcup_{i_{n-k}=i_{n-k-1}+1}^n\cdots
\bigcup_{i_2=i_{1}+1}^{n-k+2}\bigcup_{i_{1}=1}^{n-k+1}\Bigl\{(x_1,\ldots  
, x_n) : \\
x_{i_l}=\pm\frac{1}{2}\text{ for } 1\leq l\leq(n-k) \text{
  and } |x_j|\leq\frac{1}{2} \text{ if } j\neq i_l\Bigr\}/\!\sim.
  \end{multline}
So $Y_k$ is the set of all $n$-tuples where at least $n-k$ coordinates
are exactly $\pm\frac{1}{2}$. Note that $Y_k$ is the union of
${n\choose k}$ copies 
of $\bT^k$, whose pairwise intersections are all $\bT^{k-1}$. The union of
all the pairwise intersections is $Y_{k-1}$. Now by induction on $k$
we will show that $K_G^0(Y_k)$ is given by  
$${n\choose k}R\oplus{n\choose {k-1}}R/J\oplus{n\choose
  {k-2}}R\oplus\cdots\oplus{n\choose 1}R/J\oplus R,$$ 
if $k$ is even, and is given by
$${n\choose k}R/J\oplus{n\choose {k-1}}R\oplus{n\choose
  {k-2}}R/J\oplus\cdots\oplus{n\choose 1}R/J\oplus R,$$ 
if $k$ is odd. We will also show that $K_G^1(Y_k)=0$ in both cases.

Note that $Y_0=\pt$, so $K_G^*(Y_0)=R$, all in degree $0$. Let us now
look at the case of $k=1$. $Y_1$ is the one-point union 
of ${n\choose 1}=n$ $1$-tori. The point of intersection,
$y=(\pm\frac{1}{2}, \pm\frac{1}{2},\ldots , \pm\frac{1}{2})$, is a closed
$G$-invariant subset of $Y_1$, so by \eqref{eq:clgsub} we get an exact
sequence  
$$K_G^0(Y_1\backslash\{y\})\to K_G^0(Y_1)\to K_G^0(\pt)\to
K_G^1(Y_1\backslash\{y\})\to K_G^1(Y_1)\to K_G^1(\pt).$$ 
$Y_1\backslash\{y\}$ is the disjoint union of $n$ copies of
$\bR$. Using this and the above exact sequence we can see immediately
that $K_G^1(Y_1)=0$ since both $K_G^1(\bR)$ and $K_G^1(\pt)$ are
$0$ by \eqref{eq:kgrn}. And we get a short exact sequence 
\[
\xymatrix{
0\ar[r] & n(R/J)\ar[r] & K_G^0(Y_1)\ar[r] & R\ar[r] & 0,\\
}
\]
which splits since $R$ is free. So $K_G^0(Y_1)=n(R/J)\oplus R.$

Now let us look at the inductive step. Assume the above formula for
$K_G^*(Y_k)$ holds for all $k<m$. Let us also assume $m$ is
even. $Y_{m-1}$ is a closed $G$-invariant subset of
$Y_m$. $Y_m\backslash Y_{m-1}\cong {n\choose m}\bR^m$, since it is the set of all $n$-tuples with $n-m$ components exactly $\pm\frac{1}{2}$ and $m$ components with absolute value strictly less than $\frac{1}{2}$. Therefore
$K_G^*(Y_m\backslash Y_{m-1})\cong {n\choose m}R$ all in degree $0$, by \eqref{eq:kgrn} since $m$ is even. Also, by the
inductive assumption and since $m-1$ is odd,
$K_G^*(Y_{m-1})\cong{n\choose {m-1}}R/J\oplus{n\choose
  {m-2}}R\oplus{n\choose {m-3}}R/J\oplus\cdots\oplus{n\choose
  1}R/J\oplus R$ all in degree zero. By \eqref{eq:clgsub} we see that $K_G^1(Y_m)=0$ and we get a short exact sequence:
  $$0\to K_G^0(Y_m\backslash Y_{m-1})\to K_G^0(Y_m)\to K_G^0(Y_{m-1})\to 0.$$
By Proposition \ref{prop:extRJbyR} of the Appendix, the exact sequence
splits and we see that 
\begin{multline*}
K_G^*(Y_m)\cong K_G^*(Y_m\backslash Y_{m-1})\oplus
K_G^*(Y_{m-1})\\\cong {n\choose m}R\oplus{n\choose
  {m-1}}R/J\oplus{n\choose {m-2}}R\oplus\cdots\oplus{n\choose
  1}R/J\oplus R,
\end{multline*}
all in degree $0$. The inductive step for $m$ odd follows the same
form. 
Note that $Y_n$ is the entire space $\bT^n$, so by the above inductive proof we have shown:
\begin{equation}
\label{eq:kgtn}
KO_G^0(\bT^n)\cong  \left \{ \begin{array}[pos]{l} R\oplus{n\choose {n-1}}R/J\oplus{n\choose
  {n-2}}R\oplus\cdots\oplus{n\choose 1}R/J\oplus R, \text{ if $n$ is even} \\
R/J\oplus{n\choose {n-1}}R\oplus{n\choose
  {n-2}}R/J\oplus\cdots\oplus{n\choose 1}R/J\oplus R, \text{ if $n$ is odd.} \\
\end{array} \right. 
\end{equation}
$$\cong\sum_{j\leq n\text{ even}}{n\choose j}R\oplus\sum_{j\leq n,\text{ odd}}^n{n\choose j}R/J$$
Note that $0=(1-1)^n=\sum_{j\leq n\text{ even}}{n\choose j}-\sum_{j\leq n\text{ odd}}{n\choose j},$ which implies 
$$\begin{array}{rl}2^n & =(1+1)^n=\sum_{j=0}^n{n\choose j}\\
 & =\sum_{j\leq n\text{ even}}{n\choose j}+\sum_{j\leq n\text{ odd}}{n\choose j}\\
 & =2\sum_{j\leq n\text{ even}}{n\choose j}.\\
\end{array}$$
Therefore we see that $\sum_{j\leq n\text{ even}}{n\choose j}=\sum_{j\leq n\text{ odd}}{n\choose j}=2^{n-1}$. Putting this into \eqref{eq:kgtn} gives us our final result:
$$K_G^*(\bT^n)\cong 2^{n-1}\cdot R \oplus 2^{n-1}\cdot (R/J),$$
all in degree zero.
\end{proof}

When $n$ is even, Theorem \ref{thm:equivKTn} classifies stable 
D-branes on the orbifold $\bT^n/\bZ_2$ in the two type II theories. 
When $n$ is odd, the action of
$\bZ_2$ on $\bT^n$ reverses orientation, so we cannot define an
oriented string theory on $\bT^n/\bZ_2$. In order to get a consistent
string theory we would also have to mod out by the action of the
worldsheet parity operator to obtain unoriented strings. Stable
D-brane configurations would then be classified by
$KR$-theory. \cite{Witten:1998} 

Returning to the case of $n=4$, the orbifold limit of $K3$, we find
\begin{equation}
K_G^0(\bT^4)=R^8\oplus(R/J)^8,
\end{equation}
which as an abelian group is $\bZ^{24}$. This result is consistent
with the Localization Theorem (Theorem \ref{thm:localization}), which
says that $K_G^0(\bT^4)_{(J)}\cong  K_G^0((\bT^4)^G)_{(J)} \cong
(K_G^0(\pt)^{16})_{(J)}$, since
$R_{(J)}\cong (R/J)_{(J)}\cong K_G(\pt)_{(J)}$.
Note that the equivariant $K$-theory of $\bT^4$ is isomorphic to the $K$-theory
of $K3$ as an abelian group, but has the added benefit of  an
$R$-module structure on it. Using the $R$-module 
structure of $K_G^*(\bT^4)$ as well as some facts about
$G$-equivariant $K$-theory and the 
homological algebra of $R(G)$ given in the Appendix, we can now
determine the $R$-module structure on $K^*_G(M,\partial M) \cong
K^*(N,\partial N)$.  The following is a refinement of the results of
Section \ref{sec:KK3}, and can be viewed as the main mathematical
result of this paper. 
\begin{theorem}
\label{KGintM}
As before, let $G=\bZ_2$ act on $\bT^4=\bR^4/\bZ^4$ by multiplication
by $-1$ on $\bR^4$, and let $M$ be the result of removing $16$ open
balls from $\bT^4$, one ball around each fixed point of the $G$-action. 
Then \begin{equation}
\label{eq:KGintM}
K^0_G(M, \partial M)\cong (R/I)^7,\quad K^1_G(M,\partial M)\cong
(R/I)^{10}\oplus (R/2I)^5
\end{equation}
as $R$-modules.
\end{theorem}
\begin{proof}
With notation as before, $(N,\partial N)$ is obtained from
$(M,\partial M)$ by dividing out by a free 
$\bZ_2$-action, so by Theorem \ref{thm:localization}
\begin{equation}
K_G^*(M,\partial M)\cong K^*(N,\partial N)
\end{equation}
as $\bZ$-modules. Furthermore, Theorem \ref{thm:localization} tells us
that $K_G^*(M,\partial M)$ localized at $J$ must be zero. This along
with the result $K^0(N, \partial N)\cong\bZ^7$ as a $\bZ$-module, found
in the previous section, shows that as an $R$-module,
\[
K^0_G(M, \partial M)\cong (R/I)^7,
\]
by using Propositions \ref{prop:RstrZ} and \ref{prop:RstrZZ} from the
Appendix. Alternatively, the $R$-module structure on $K^0_G(M,
\partial M)$ is determined by the action of $t$, which by Proposition
\ref{prop:equivKfree} from the Appendix is given by cup-product with
the class of a line bundle in $K^0(N)$, whose first Chern class is
torsion. Since $K^0(N, \partial N)$ is torsion-free, this action has
to be trivial.

Since $K^1_G(M,\partial M)\cong K^1(N, \partial
N)\cong\bZ^{15}\oplus\bZ_2^5$ as a $\bZ$-module, and must be zero when
localized at $J$, it is an extension of $(R/I)^{15}$ by $(R/(2,I))^5$, and
the possible $R$-module structures are: 
\begin{equation}
\label{eq:K1Grel}
(R/I)^{15-k}\oplus(R/(2,I))^{5-k}\oplus(R/2I)^k,
\end{equation}
where $0\leq k\leq 5$. There are a few ways to 
determine the value of $k$.  One method is to use the long exact
sequence induced on equivariant $K$-theory for the pair $(\bT^4,\partial M)$:
\begin{equation}
\label{eq:pairT4bdM}
\cdots \to K_G^0(\bT^4,\partial M)\xrightarrow{\alpha} K_G^0(\bT^4)
\xrightarrow{\beta} K_G^0(\partial M) \xrightarrow{\gamma}
K_G^1(\bT^4,\partial M)\to K_G^1(\bT^4)\to\cdots. 
\end{equation}
By excision, 
\begin{align}
K_G^*(\bT^4,\partial M) &\cong K_G^*(\bT^4-\partial M)\notag\\ &\cong
K^*_G(\text{int }M)\oplus K_G^*(16\bR^4)\notag\\ &\cong  K^*_G(M,
\partial M)\oplus (K_G^*(\bR^4))^{16}.
\end{align}
$K_G^*(\bR^4)$ is given by equation \eqref{eq:kgrn} and is all in
degree zero, so
\[
K_G^1(M,\partial M)\cong K_G^1(\bT^4, \partial M).
\]
The $R$-module
structure of $K_G^*(\partial M)\cong K_G^*(16S^3)$ is given in Example
\ref{ex:S3}. Plugging this, as well as $K_G^*(\bT^4)$ from Theorem
\ref{thm:equivKTn}, into the long exact sequence \eqref{eq:pairT4bdM},
we get the exact sequence
\begin{equation}
\label{eq:LESR}
0\to(R/I)^{16}\to R^{16}\oplus(R/I)^7\xrightarrow{\alpha}
R^8\oplus(R/J)^8\xrightarrow{\beta} (R/2I)^{16}\xrightarrow{\gamma}
K_G^1(M,\partial M)\to 0
\end{equation}
or
\[
0 \to (\coker \alpha \cong \im \beta) \to (R/2I)^{16}\xrightarrow{\gamma}
K_G^1(M,\partial M)\to 0.
\]
Now $\im \beta$ contains the
diagonal copy $\Delta(R/2I)$ of $R/2I$ inside $(R/2I)^{16}$. One can
see this as follows: the trivial line bundle $1_{\bT^4}$ on $\bT^4$ restricts
to the trivial line bundle on each component of $\partial M$, or
in other words to $\Delta(\dot 1)$, where $\dot 1$ is the image of $1$
in $R/2I$. Since $R/2I$ is generated as an $R$-module by $\dot 1$
and $r\cdot 1_{\bT^4}$ restricts to $\Delta(r\cdot \dot 1)$ for $r\in
R$, $\Delta(R/2I) \subseteq \im \beta$. Since
$(R/2I)^{16} /\Delta(R/2I) \cong (R/2I)^{15}$, $K_G^1(M,\partial M)$
is a quotient of $(R/2I)^{15}$. Therefore it
has at most $15$ cyclic summands as an $R$-module, so
$(15-k)+(5-k)+k=20-k\le 15$ and $k=5$, which completes the proof.
\end{proof}

Since $K^*(N)$ is used to classify $D$-brane charges in the type IIB
theory, we can also  compute the absolute (i.e., 
not rel boundary) equivariant $K$-theory of $M$.
\begin{theorem}
\label{prop:KGM}
Using the same notation as above
\begin{equation}
K^0_G(M)\cong (R/2I) \oplus (R/(2,I))^4\oplus (R/I)^6,\quad K^{-1}_G(M)\cong
(R/I)^{15}
\end{equation}
as $R$-modules.
\end{theorem}
\begin{proof}
We start with equation \eqref{eq:KN}, which determines $K^*_G(M)$ as a
$\bZ$-module. Just as in the proof above, we immediately conclude that
$K^{-1}_G(M)\cong (R/I)^{15}$ and that 
\[
K^0_G(M)\cong (R/I)^{7-k}\oplus(R/(2,I))^{5-k}\oplus(R/2I)^k,
\]
where $0\leq k\leq 5$. To finish the proof, we can use Proposition
\ref{prop:equivKfree} of the Appendix. Since $N=M/G$ has all its
even-dimensional cohomology in degrees $0$ and $2$, the Chern
character is a ring isomorphism $K^0(N) \to \bZ \oplus H^2(N,\bZ)$,
with $H^2(N,\bZ)$ an ideal whose square is zero, and
the action of $t$ corresponds to multiplication by $1+c$, where $c\in
H^2(N,\bZ)$ corresponds to the $2$-fold covering map $M\to M/G$. 
Note that $\bZ\cdot 1 + \bZ_2\cdot c$ is a subring of $K^0(N)$
isomorphic to $R/2I$, and Proposition \ref{prop:equivKfree} says that
the $R$-module structure on $K^0_G(M)$ comes from the action of this
subring. However, one can
also see that it is impossible to get two elements of $K^0_G(M)$, each
with annihilator $2I$, which are linearly independent over $R/2I$. For
suppose such elements had the form $n_j+x_j$, with $n_j\in \bZ$ and
$x_j\in H^2(N,\bZ)$. The condition that the annihilator of $n_j + x_j$
is $2I$ means
$n_j$ is odd. But then $n_1+n_2$ is even, so the sum of the two elements
is annihilated by $I$, and so they are not linearly independent over
$R/2I$. So $K^0_G(M)$ cannot have more than one summand isomorphic to $R/2I$.
\end{proof}


\section{Conclusion}
\label{sec:concl}

We have seen that the group that classifies the BPS D-branes in type-I
superstring theory compactified on $\bT^4$, $KO(\bT^4\times \bR^6)$,
is isomorphic to $\bZ^6\oplus {\bZ_2}^5$, which injects into both
$K^{-1}\bigl((N, \partial N)\times\bR^6\bigr) \cong \bZ^{15}
\oplus{\bZ_2}^5$ and $K^0(N\times\bR^6)$ with an isomorphism on the
torsion. The isomorphism on torsion is
significant because the torsion brane charges occurring in the type I
theory  can only be explained by $D$-branes carrying torsion charge in
the type II theories and could not be explained by the $K$-theory of
$K3$. There would also be no way to obtain torsion charges from charge
groups for R-R or NS-NS fields living in cohomology of $K3$, since
once again this is torsion-free. To understand the difference in the free ranks, we must first understand which integral brane charges we expect to correspond to stable $D$-branes in the different theories. We present two possible avenues of
future research to explain the differences in the free rank. 

 One possibility is that we could need to include extra $\bZ$ summands
 that do not appear in $KO^{-6}(\bT^4)$ corresponding to other
 non-supersymmetric branes in type I on 
$\bT^4$ (such as those discussed in \cite{MR1731758}), since these would have 
$K$-theoretic charges living in the other groups listed in Table
\ref{table:KOgroups}. For example, D$6$-branes with real Chan-Paton
bundles should have charges living in $KO^{-3}(\bT^4)\cong \bZ^4$. It
is unknown which non-BPS branes in the type I theory transform to BPS
branes in the type IIA theory. Further research into this phenomenon
needs to be done before completing the classification. A correct
classification of the BPS brane charges on the type IIA side is a
powerful tool when studying non-BPS branes that map to BPS
branes. If we knew all of the stable brane charges in the type IIA
theory, it would give an upper bound on the additional stable non-BPS
branes we must include in our classification of type I $D$-branes by
the duality. It is important to note that it would not be enough to
just know all of the stable brane charges in the type IIA theory, but
we would also need to know which branes carry RR charge and which
carry NS-NS charge. However, if we knew all of the stable brane
charges and had a correct classification of the $D$-brane charges  
via $K$-theory, by elimination we would know the stable NS-NS charges
as well. 

To understand which $D$-branes should be included in the
classification, we believe the next step is to determine exactly what
cycles a $D$-brane can wrap. In the literature, when
analyzing different aspects of this duality, different authors have
used both  a smooth $K3$ and the orbifold limit, $\bT^4/\bZ_2$. In our
calculation we 
combined different features from a smooth $K3$ and $\bT^4/\bZ_2$. By
looking at features from both points in the moduli space we were able
to describe the torsion brane charges successfully. To fully
understand the discrepancy in the free rank, one needs to determine
which cycles come from which description and which ones
should be included in the classification. 

The type IIA theory on $K3$ has $24$ gauge bosons, all coming from the
RR sector. $22$ of them correspond to reductions of the type IIA RR
$3$-form potential on $2$-cycles of $K3$. Of these, $19$ correspond to
anti-self-dual forms, and $16$ of those map to the $16$ gauge bosons
in the heterotic theory that come from the Cartan subalgebra of the
rank-$16$ gauge group \cite[p.\ 11]{Aharony:2003}. In the type IIA
theory at the orbifold limit these $16$ gauge bosons correspond to
$D2$-branes wrapping $2$-cycles at the singularities. The $2$-cycles
correspond to the blow-up of the orientifold fixed points and go to
zero size at the orbifold point. While the sizes of the $2$-cycles go
to zero, the $D$-brane tensions remain finite \cite{Majumder:1999}.
These $16$ $D2$-branes are supersymmetric since the cycles they wrap
are. When the radius of the circle of the original torus that passes
through two fixed points is smaller than a certain critical radius,
the two BPS $D2$-branes associated to the fixed points will decay into
a single stable non-BPS $D2$-brane. The single non-BPS $D2$-brane can
be viewed as wrapped around a single $2$-cycle of $K3$ that is
homologically equivalent to the sum of the supersymmetric $2$-cycles
corresponding to the two fixed points \cite{Majumder:1999}. In this
region the non-BPS brane has a lower mass than that of the combined
system of two BPS branes, so the non-BPS brane configuration is the
stable one. Performing a $T$-duality transformation sends these
$D2$-branes to $D$-strings between the different $NS5$-branes in the
type IIB theory, and sends the blow-up modes that determine the $2$-cycles in
the type IIA case to moduli determining the location of the
$NS5$-branes \cite{Majumder:1999, Aharony:2003}. The locations of the
$NS5$ branes affect the mass of the $D$-string  configurations
(and hence the $D2$-branes in the IIA theory, by $T$-duality), since
changing the locations of the $NS5$-branes changes the lengths of the
strings \cite{Majumder:1999}. In our computation of stable $D$-brane
charges in the type IIB theory, we moved the $NS5$-branes away from the
orbifold fixed points, changing the mass of the BPS brane
system. However, in the type IIA theory we looked at the interior of
$N$. When we excise the fixed points the $2$-cycles that shrink to
size zero at the orbifold points can no longer shrink to zero size in
$\mathring N$. In $\mathring N$ there do exist $2$-cycles that are the
sums of pairs of the size zero $2$-cycles with a factor of a half as
can be seen by looking at the Kummer lattice. More research needs to
be performed to determine if we want to include the charges from the
BPS branes or the non-BPS branes. Note that since a single non-BPS
brane corresponds to 2 BPS branes, this could possibly reduce the free
rank of the $K$-theoretic classification of $D$-brane charges in the
two type II theories. The $R(G)$-module structure of
$K_{\bZ_2}^*(M,\partial M)$ and $K^*_{\bZ_2}(M)$ provides added
constraints (beyond those coming
from the abelian group structure) that can be used to complete this
determination, since it provides added information about where
different charges in $N$ come from in relation to the orbifold
$\bT^4/\bZ_2$. 

Another benefit of Theorems \ref{thm:equivKTn} and \ref{KGintM}
becomes apparent when trying to include a twisting due to an
$H$-flux. On an orbifold $X/G$, it is unclear what is meant by the
$H$-flux. It does not make sense for $H$ to live in $H^3(X;\bZ)$ since
to make any sense on the orbifold, $H$ would have to be $G$
invariant. Recent work by Distler, Freed and Moore in
\cite{Distler:2009} proposes using more exotic twistings involving
equivariant cohomology, but the precise definition of the $H$-flux on
an orbifold remains an open problem. Knowing the equivariant
structures of $K_{\bZ_2}^*(M,\partial M)$ and $K_{\bZ_2}^*(\bT^n)$
allows us to apply twistings that take advantage of the equivariant
structure. Furthermore, when $n=4$, Theorem \ref{thm:equivKTn} shows
that considering $\bZ_2$-equivariant states on the orbifold limit of
$K3$ gives a $D$-brane spectrum equivalent to the one arising from
considering states 
on $K3$ itself. Note that there can be no non-trivial $H$-flux on $K3$
(since $H^3(K3)=0$), but it may be possible to apply a twisting such
as one of the ones proposed in \cite{Distler:2009} or \cite{Braun:2002} to $K_{\bZ_2}^*(\bT^4)$. 

While showing that the stable D-brane configurations in the two theories
match does not prove the two theories are dual, this particular
example illustrates how ensuring that stable D-brane configurations 
match is a useful first step in checking a possible duality. By
looking at the stable D-branes in this case we saw immediately that we
did not want to use $K3$, but rather a desingularized version of the
orbifold blow-down, as matched with our physical
intuition. Furthermore, the $K$-theoretic classification of $D$-branes
provides useful tools for studying other phenomena that arise in
string theory dualities, such as non-BPS branes that map to BPS
branes.  

This example also illustrates the benefit of composing known dualities
to gain information about less understood dualities. By
composing a known duality between the type I and $\SO(32)$ heterotic
theories with one between the $\SO(32)$ heterotic and type IIA
theories to obtain a duality between the type I and type IIA theories,
and then composing that with $T$-duality between type IIA and
IIB we were able to come up with an examples of dualities between all
$5$ superstring theories. Knowing what the theory looks like in all of
the superstring theories provides more information to use when testing
possible dualities.  


\section{Appendix: Some Homological Algebra over $R(\bZ_2)$}
\label{sec:appendix}

This appendix collects together some facts about the homological
algebra of the representation ring $R=R(G)$ for $G=\bZ_2$, which are
useful for studying $G$-equivariant $K$-theory. The connection is that
$G$-equivariant $K$-groups are always modules over $R=R(G)$, and the
$R$-module structure carries more information than just the abelian
group structure of the $K_G$-groups. The basic reference for
representation rings is \cite{MR0248277} and the basic reference for
equivariant $K$-theory is \cite{MR0234452}.

Recall that for a compact group $G$, $R(G)$ is the free abelian group
on the equivalence classes of irreducible (finite-dimensional complex)
representations of $G$, with multiplication coming from the tensor
product of representations. When $G$ is also abelian, $R(G)$ is just
the group ring of the Pontrjagin dual group $\widehat G$.  In what
follows, we always take $G=\bZ_2$, $R=R(G)=\bZ[t]/(t^2-1)$. The
generator $t$ corresponds to the nontrivial character of $G$. The ring
$R$ has two important prime ideals, $I=(t-1)$ (the augmentation ideal)
and $J=(t+1)$. These play symmetrical roles since there is an
automorphism of $R$ (not coming from an automorphism of $G$)
which interchanges them. In the sense of
\cite{MR0248277}, $I$ has support $\{0\}$ (i.e., just the identity
element) and $J$ has support $G$.  If we localize $R$ at $I$, we get a
local ring $R_{(I)}$ in which everything not in $I$ is invertible. In
particular, every prime $p\in \bN$ is inverted, so $R_{(I)}$ is a
$\bQ$-vector space. But $R\otimes \bQ$ splits as a direct sum
$\bQ\oplus \bQ$, with $t$ acting by $1$ on one factor and by $-1$ on
the other factor. Since $t+1$ must be invertible on $R_{(I)}$, the
summand where $t$ acts by $-1$ must die and so $R_{(I)}\cong \bQ$ with
$t$ acting by $+1$. Similarly, $R_{(J)}\cong \bQ$ with
$t$ acting by $-1$. Furthermore, we have $I\cdot J = 0$, $I =
\Ann_R(J)$, $J=\Ann_R(I)$, $R/I\cong J$ (as $R$-modules), and
$R/J\cong I$ (as $R$-modules). Also, $(R/J)_{(I)}=0$, $(R/I)_{(J)}=0$,
$(R/J)_{(J)}\cong R_{(J)}$, and $(R/I)_{(I)}\cong R_{(I)}$.  But $I+J$
is a proper $R$-submodule of $R$ (of index $2$).

The Segal Localization Theorem \cite[Proposition 4.1]{MR0234452}
specializes to the following:
\begin{theorem}[Segal]
\label{thm:localization}
Let $G=\bZ_2$ and let $X$ be a locally compact $G$-space. Then
$K^*_G(X)_{(J)}\!\cong  K^*(X^G)\otimes_{\bZ}\bQ$, with $t$ acting by
$-1$. In particular, if $G$ acts freely on $X$, then
$K^*_G(X)_{(J)}=0$, and $K^*_G(X)\cong K^*(X/G)$, at least as $\bZ$-modules.
\end{theorem}

When $G$ acts freely on $X$, one can make Theorem
\ref{thm:localization} a bit more precise.
\begin{proposition}
\label{prop:equivKfree}
Let $G=\bZ_2$ and let $X$ be a compact {\bfseries free} $G$-space. 
Then the $R(G)$-module structure on $K^*_G(X)\cong K^*(X/G)$ is
defined by letting $t$ act by tensoring with the line bundle $V$
with $c_1(V)=c$, where $c$ is the image in $H^2(X/G, \bZ)$ under the
Bockstein homomorphism of the class in $H^1(X/G, \bZ_2)$ classifying
the $2$-to-$1$  covering map $X\to X/G$. 
{\lp}If $X$ is a connected reasonable space, such as a manifold, then
$H^1(X/G, \bZ_2) \cong \Hom (\pi_1(X/G), \bZ_2)$ classifies $2$-to-$1$
covering  spaces of $X/G$, by covering space theory.
One can also realize $V$ more explicitly as the fiber product
$X\times_G \bC$, 
where $G$ acts on $\bC$ by the nontrivial character $t$.{\rp}

If $A$ is a closed $G$-invariant subspace of $X$, then the
$R(G)$-module structure on 
\[
K^*_G(X,A)\cong K^*(Y, B), \quad Y=X/G,\, B=A/G,
\]
is again
defined by letting $t$ act by cup-product with $[V]\in
K^0(X/G)$. {\lp}Recall that for any pair $(Y,B)$, we have the cup-product
$K^0(Y)\otimes K^*(Y,B) \to K^*(Y,B)$.{\rp}
\end{proposition}

\begin{proof}
The definition of the $R(G)$-action on $K^*_G(X)$ or on $K^*_G(X,A)$
implies that the result of applying the module action of $t$
corresponds to tensoring with $(\bC, t)$, which is the same after
applying the isomorphisms $K^*_G(X)\cong K^*(Y)$ or $K^*_G(X,A)\cong
K^*(Y,B)$ as taking the vector bundle tensor product with $V$. The rest is
immediate. 
\end{proof}

\begin{corollary}
\label{cor:KGtriv}
If $G=\bZ_2$ and $X$ is a compact free $G$-space, then the
$R$-module structure on $ K^*(X/G)$ is trivial {\lp}i.e., factors
through $R/I${\rp} if and only if $c=0$ in $H^2(X/G,\bZ)$ {\lp}in
the notation of the Proposition{\rp}.
\end{corollary}
\begin{proof}
If $c=0$ in $H^2(X/G,\bZ)$, then (in the notation of Proposition
\ref{prop:equivKfree}) $V$ is the trivial bundle and the action of $t$
is trivial. But if $c\ne 0$, then $t\cdot 1 = [V]\ne 1$, so the
action is nontrivial.
\end{proof}

\begin{example}
\label{ex:S3}
Let $S^3$ be given the antipodal action of $G$. Then
  $K^0_G(S^3) \cong R/2I$ and $K^1_G(S^3) \cong R/I$. Indeed, we know
  that $K^0_G(S^3) \cong 
  K^0(S^3/G) = K^0(\bR\bP^3) \cong \bZ\oplus \bZ/2$ as an abelian
  group, but by Corollary \ref{cor:KGtriv}, the $R$-module structure must be
  nontrivial, so this follows from Proposition \ref{prop:RstrZZtwo}
  below. The $R$-module structure on $K^1_G(S^3) \cong 
  K^1(S^3/G) = K^1(\bR\bP^3) \cong \bZ$ must be trivial by
  Proposition \ref{prop:RstrZ} below. Alternatively, one can compute
  directly from the exact   sequence
\[
0 = K^1_G(D^4) \to 
K^1_G(S^3)\to K^0_G(\text{int}\,D^4) \xrightarrow{\beta}
 K^0_G(D^4) \to K^0_G(S^3)
\to K^1_G(\text{int}\,D^4)=0.
\]
We have $K^0_G(D^4) \cong K^0_G(\pt) = R$ since $D^4$ is equivariantly
contractible, and $K^0_G(\text{int}\,D^4) \cong R$ by equivariant Bott
periodicity. So one only needs to compute the map $\beta$, which is
multiplication by the Bott element. This is the exterior algebra
complex of $(\bC,t)^2$, which is $2(1-t)$. So the image of $\beta$ is
exactly $2I$, and the kernel of $\beta$ is $\Ann_R(2I)=J\cong R/I$.
\end{example}

We now need some facts about certain special $R$-modules.
\begin{proposition}
\label{prop:RstrZ}
Let $M$ be an $R$-module which as a $\bZ$-module is isomorphic to
$\bZ$. Then either $M\cong R/I$ or $M\cong R/J$.
\end{proposition}
\begin{proof}
The $R$-module structure is determined by the action of $t$, which
must be an automorphism of $M$ as a $\bZ$-module. Since
$\Aut_{\bZ}(\bZ) = GL(1,\bZ) = \{1, -1\}$, there are exactly two
possibilities: $R/I$ if $t$ acts by $+1$, and $R/J$  if $t$ acts by $-1$.
\end{proof}

\begin{proposition}
\label{prop:RstrZtwo}
Let $M$ be an $R$-module which as a $\bZ$-module is isomorphic to
$\bZ_2$. Then $M\cong R/(2,I)\cong I/2I$.
\end{proposition}

\begin{proof}
The $R$-module structure is determined by the action of $t$, which
must be an automorphism of $M$ as a $\bZ$-module. But $\bZ_2$ has no
non-trivial automorphisms, so there is only the trivial possibility.
\end{proof}

\begin{proposition}
\label{prop:RstrZZ}
$\Ext_R^1(R/I,R/I)=0$, and thus every $R$-module extension of $R/I$ by
$R/I$ splits. Similarly, $\Ext_R^1(R/J,R/J)=0$, so every $R$-module
extension of $R/J$ by $R/J$ splits. 
\end{proposition}

\begin{proof}
We only do the first case, as the second is precisely analogous. Start
with the extension of $R$-modules $0\to I \to R \to R/I\to 0$, and
apply $\Hom_R(\underline{\phantom{X}},R/I)$. We get
\begin{multline*}
0 \to \Hom_R(R/I,R/I) \xrightarrow{\alpha} \Hom_R(R,R/I) \\
\to \Hom_R(I,R/I) \to 
\Ext_R^1(R/I,R/I) \to \Ext_R^1(R,R/I)=0.
\end{multline*}
Now $\Hom_R(R,R/I)\cong R/I$, and since any $R$-module map $R\to R/I$
is determined by the image of $1$, which is annihilated by $I$, it
comes from something in $\Hom_R(R/I,R/I)$. So the map $\alpha$ is an
isomorphism, and $\Ext_R^1(R/I,R/I)\cong \Hom_R(I,R/I)$. Since
$R/I\cong J$ (this follows from Proposition \ref{prop:RstrZ}), a 
homomorphism $I\to R/I$ is the same thing as a homomorphism
$\varphi\co I\to J$,
which is determined by the image $\varphi(t-1)$ of $t-1$. But $t-1$ is
annihilated by $t+1$, so $(t+1)\varphi(t-1)=0$. But the only element
of $J$ annihilated by $t+1$ is $0$, so $\varphi=0$.
\end{proof}

\begin{proposition}
\label{prop:extRJbyR}
$\Ext_R^1(R/J,R)=0$, and thus every $R$-module of extension of $R/J$
by $R$ splits. In addition, $\Hom_R(R/J, R)\cong R/J$. Similarly,
$\Ext_R^1(R/I,R)=0$, so every $R$-module of 
extension of $R/I$ by $R$ splits, and   $\Hom_R(R/I, R)\cong R/I$.
\end{proposition}

\begin{proof}
We only do the first case of $R/J$, as the case of $R/I$ is precisely
analogous. First observe (as in the proof of Proposition
\ref{prop:RstrZZ}) that $\Hom_R(R/J, R)\cong I\cong R/J$, since a
homomorphism $R/J\to R$ is completely determined by the image of the
coset of $1$, which can be anything in $\Ann_R(J)=I$.
Consider the short exact sequence  
$$0\to J\to R\to R/J\to 0,$$ 
and apply the functor $\Hom_R (\underline{\phantom{X}}, R)$ to it. We get
\begin{equation}
\label{eq:ext}
\begin{aligned}
0 &\to \Hom_R(R/J, R) \to \Hom_R(R, R) \to \Hom_R(J, R) \\ &\to 
\Ext^1_R(R/J,R) \to \Ext^1_R(R,R)=0.
\end{aligned}
\end{equation}
We have $\Hom_R(R, R)=R$, and by the above, $\Hom_R(J, R)=J\cong
R/I$. Similarly, $\Hom_R(R/J, R)=\Hom_R(I, R)=I$. So 
\eqref{eq:ext} becomes the sequence 
$$0 \to I \to R \to R/I \to 
\Ext^1_R(R/J,R) \to 0,$$
and $\Ext^1_R(R/J,R) =0$.
\end{proof}
\begin{proposition}
\label{prop:extRIbyRJ}
If $M$ is an $R$-module that sits in a short exact sequence
\[
0 \to R/I \to M\to R/J \to 0 \text{ or }
0 \to R/J \to M\to R/I \to 0,
\]
then either $M\cong R$ or else $M\cong R/I \oplus R/J$.
\end{proposition}
\begin{proof}
Note that $R/I$ has the periodic free resolution
\[
\cdots \xrightarrow{t+1} R \xrightarrow{t-1} R \xrightarrow{t+1} R 
\xrightarrow{t-1} R \to R/I \to 0.
\]
Thus $\Ext^*_R(R/I,R/J)$ is the cohomology of the complex
\[
R/J \xrightarrow{-2} R/J \xrightarrow{0} R/J \xrightarrow{-2} R/J \to \cdots,
\]
and $\Ext^1_R(R/I,R/J)\cong R/(2,J)$, which has only two
elements. Thus there can be at most two isomorphism classes of
extensions of $R/I$ by $R/J$. Similarly for extensions of $R/J$ by
$R/I$. Since $R$ and $R/I \oplus R/J$ are clearly nonisomorphic
extensions of the desired form, they are the only possibilities.
\end{proof}
\begin{proposition}
\label{prop:RtwoI}
The $R$-module extension
\[
0\to I/2I \to R/2I \to R/I \to 0
\]
does {\bfseries not} split.
\end{proposition}
\begin{proof}
Suppose to the contrary that we have a splitting $s\co R/I \to
R/2I$. Then $s$ is determined by $s(1+I)$, which let us say is $a+bt +
2I$. Since $s(1+I)$ projects to $1+I$, $a+bt\in 1 + I$, and can be
rewritten as $1 + c(1-t)$. But $t$ acts on $R/I$ by $+1$, so we must
have $t(1 + c(1-t))-(1 + c(1-t))\in 2I$, i.e., $(-1-2c)(1-t)\in
2I$. This is impossible since $-1-2c$ is odd, not even.
\end{proof}

\begin{proposition}
\label{prop:RstrZZtwo}
$\Ext_R^1\bigl(R/I,R/(2,I)\bigr)\cong R/(2,I)$. Thus any $R$-module
extension of $R/I$ by $R/(2,I)$ either splits or is isomorphic to
$R/2I$. The same holds with $I$ replaced everywhere by $J$.
\end{proposition}

\begin{proof}
\textbf{First argument.} Since any $\bZ$-module extension of $\bZ$ by
$\bZ/2$ splits, the issue is to compute the possible $R$-module
structures on $\bZ\oplus \bZ_2$, which mod $\bZ$-torsion give the
$R$-module $R/I$. Such structures are determined by the action of
$t$. If it is trivial, the module splits as $(R/I)\oplus
(R/(2,I))$. If it isn't trivial, then $t$ must map $(1,\bar0)$ to
$(1,\bar1)$, where $\bar0$ and $\bar1$ are the two elements of
$\bZ_2$. This is precisely the situation in the non-split extension
$R/2I$ of Proposition \ref{prop:RtwoI}, since $t\cdot 1 = 1 - (1-t)$,
and $1-t$ represents the nontrivial class in $I/2I$. \qedsymbol

\textbf{Second argument.} Alternatively, we can compute $\Ext$
directly, following the method used in the proof of Proposition
\ref{prop:RstrZZ}. Consider the exact sequence
\begin{multline*}
0 \to \Hom_R(R/I,R/(2,I)) \xrightarrow{\alpha} \Hom_R(R,R/(2,I)) \\
\to \Hom_R(I,R/(2,I)) \to 
\Ext_R^1(R/I,R/(2,I)) \to \Ext_R^1(R,R/(2,I))=0.
\end{multline*}
Since $R/(2,I)$ is $\bZ_2$ with the trivial action of $t$, it is easy
to see that the three $\Hom$ groups in this sequence are all isomorphic
to $\bZ_2$ as abelian groups, and thus to $R/(2,I)$ as $R$-modules (by
Proposition \ref{prop:RstrZtwo}). Thus the $\Ext$ group is also
isomorphic to $\bZ_2$ as an abelian group, and to $R/(2,I)$ as an $R$-module.
\end{proof}

In dealing with compact $G$-manifolds $X^n$ (possibly with boundary)
such as $\bT^4$ and $M$, it is sometimes useful to use equivariant
Poincar\'e duality, which applies whenever the manifold has an
equivariant spin$^c$-structure (as is the case for our examples).
Thus one has isomorphisms $K_G^*(X) \cong K_{n-*}^G(X, \partial X)$
and $K_G^*(X, \partial X) \cong K_{n-*}^G(X)$. At the same time,
equivariant $K$-theory and equivariant $K$-homology are related by a
universal coefficient theorem.
\begin{theorem}[Universal Coefficient Theorem \cite{MR948692}]
\label{thm:UCT}
Let $X$ be a locally compact $G$-space. Then there is a short exact
sequence of $R$-modules
\[
0\to \Ext^1_R(K^{*+1}_G(X),R) \to K_*^G(X) \to \Hom_R(K^*_G(X), R) \to
0,
\]
which doesn't necessarily split.
\end{theorem}
\begin{example}
\label{ex:S3contd}
Let's continue with Example \ref{ex:S3} about $S^3$ with the antipodal
action. Since $\dim S^3$ is odd, Poincar\'e duality gives
$R/2I\cong K^0_G(S^3)\cong K_1^G(S^3)$ and $R/I\cong K^1_G(S^3)\cong
K_0^G(S^3)$. Let's see that this is consistent with Theorem
\ref{thm:UCT}. We have $\Hom_R(R/I, R)\cong R/I$ and $\Ext_R^1(R/I,
R)=0$ by Proposition \ref{prop:extRJbyR}. A homomorphism $R/2I\to R$
must kill the $\bZ$-torsion $I/2I$ in $R/2I$, hence must factor through
$R/I$, so also $\Hom_R(R/2I, R)\cong R/I$. Finally, $\Ext_R^1(R/2I,
R)\cong R/(2,I) $. One can see this from the exact
sequence
\begin{multline*}
0\to \Hom_R(R/2I, R)\cong R/I \cong J \to
R\cong \Hom_R(R,R) \\ \to \Hom_R(2I,R) \cong I \to 
\Ext_R^1(R/2I,R) \to \Ext_R^1(R,R)  =0.
\end{multline*}
So for example we have a short exact sequence
\begin{multline*}
0\to \Ext^1_R(K^0_G(S^3),R)\cong R/(2,I) \to K_1^G(S^3) \cong R/2I \\ \to
\Hom_R(K^1_G(S^3),R) \cong R/I\to 0,
\end{multline*}
and the sequence does \emph{not} split by Proposition
\ref{prop:RtwoI}.
This resolves a question left unanswered in \cite{MR948692}.
\end{example}

\bibliographystyle{hplain}
\bibliography{K3}
\end{document}